\documentclass[journal]{IEEEtran}
\ifCLASSINFOpdf
\else
   \usepackage[dvips]{graphicx}
\fi
\usepackage{url}
\usepackage{amsmath}
\usepackage{graphicx}
\usepackage{placeins}
\usepackage{float}
\usepackage{bbm}
\usepackage{verbatim}
\usepackage{amssymb}
\usepackage{amsmath}
\usepackage{amsthm}
\usepackage{mwe}
\usepackage[dvipsnames]{xcolor}
\usepackage{mathtools}
\usepackage{bbm}
\usepackage{docmute}
\usepackage{graphicx}
\usepackage[ruled,vlined]{algorithm2e}
\let\ORIincludegraphics\includegraphics
\renewcommand{\includegraphics}[2][]{\ORIincludegraphics[scale=0.6,#1]{#2}}
\usepackage{mathtools}
\DeclarePairedDelimiter{\ceil}{\lceil}{\rceil}

\usepackage{graphicx}

\newcommand{\defeq}{\vcentcolon=}

\newtheorem{definition}{Definition}
\newtheorem{remark}{Remark}
\newtheorem{lemma}{Lemma}

\begin{document}

\title{Distributed Online Learning for Coexistence in Cognitive Radar Networks}

\author{William W. Howard, Anthony F. Martone, R. Michael Buehrer
\thanks{W.W. Howard and R.M. Buehrer are with Wireless@VT, Bradley Department of ECE, Virginia Tech, Blacksburg, VA, 24061. (e-mails:$\{$wwhoward, buehrer$\}$@vt.edu)  \\
A.F. Martone is with the U.S. Army Research Laboratory, Adelphi, MD 20783. (e-mail:{anthony.f.martone.civ}@army.mil).\\
The support of the U.S. Army Research Office (ARO) is gratefully acknowledged. \\
Portions of this work were presented at IEEE Radar Conf, Atlanta, GA, May 2021 \cite{howard2021multiplayer}. Other portions will be presented at IEEE Radar Conf, New York, NY, March 2022 \cite{howard2022adversarial}.}}





\maketitle
\pagenumbering{roman}
\begin{abstract}



This work addresses the coexistence problem for radar networks. 
Specifically, we model a network of cooperative, independent, and non-communicating radar nodes which must share resources within the network as well as with non-cooperative nearby emitters. 
We approach this problem using online Machine Learning (ML) techniques. 
Online learning approaches are specifically preferred due to the fact that each radar node has no prior knowledge of the environment nor of the positions of the other radar nodes, and due to the sequential nature of the problem.  
For this task we specifically select the multi-player multi-armed bandit (MMAB) model, which poses the problem as a sequential game, where each radar node in a network makes independent selections of center frequency and waveform with the same goal of improving tracking performance for the network as a whole.
For accurate tracking, each radar node communicates observations to a fusion center on set intervals. 
The fusion center has knowledge of the radar node placement, but cannot communicate to the individual nodes fast enough for waveform control. 
Every radar node in the network must learn the behavior of the environment, which includes rewards, interferer behavior, and target behavior. 
Each independent and identical node must choose one of many waveforms to transmit in each Pulse Repetition Interval (PRI) while avoiding \emph{collisions} with other nodes and interference from the environment. 
The goal for the network as a whole is to minimize target tracking error, which relies on obtaining high SINR in each time step. 
Our contributions include a mathematical description of the MMAB framework adapted to the radar network scenario. 
We conclude with a simulation study of several different network configurations. 
Experimental results show that iterative, online learning using MMAB outperforms the more traditional sense-and-avoid (SAA) and fixed-allocation approaches. 



\end{abstract}

\IEEEpeerreviewmaketitle
\vspace{-.37cm}
\section{Introduction}
The demand for spectrum reaches new peaks daily in this era of fifth-generation wireless technologies. 
This increased demand for spectrum has an impact on all types of radio access devices. 
Radar systems in particular are being affected and must often share their spectrum with other services. 
To address this impact, our work is concerned with Dynamic Spectrum Access (DSA) schemes for Cognitive Radar Network (CRN) systems.
DSA is an umbrella term for techniques which seek to assure spectrum access to dissimilar, and often non-cooperative, systems which attempt to access the same resource while simultaneously attempting to mitigate harmful interference between wireless devices \cite{4205091}. 
Specifically with the increased use of spectrum by commercial communication systems in the GHz bands, there is a need for radar systems to be DSA-capable as secondary users (SU). 
However, while the benefits of using a network of radars are numerous (multiple observations of a target, resiliency to physical damage, etc.), they only exacerbate the coexistence problem. 
To address this need, this work focuses on a Multi-player Multi-Armed Bandit (MMAB) approach to DSA in CRNs.

Cognitive systems are an attractive platform for the DSA problem for several reasons. 
Primarily, we leverage the ability of cognitive radar systems to \emph{monitor their environment} to inform future decisions. 
Of the two modes of CRN cognition outlined by Haykin in his seminal work on the topic \cite{1574168}, this approach falls under ``distributed cognition'', where each node in the network is a cognitive agent. 
This contrasts with the alternative ``centralized cognition'' where a \emph{central coordinator} is the only intelligent agent. 
We do, however, describe a ``fusion center'' which collects observations from the network for the purpose of target tracking on the time scale of the Coherent Pulse Interval (CPI)\footnote{For comparison's sake, one CPI contains 1024 PRIs, which is the time scale of individual decision making.}.   
Note that our MMAB technique does not facilitate learning in a joint or federated sense \cite{9060868}; each radar node is aware of the presence of the network, but makes its own decisions. 
In this network, each node is aware of its own position, and the fusion center is aware of all node locations. 
Nodes are not aware of the placement of other nodes. 
This follows the assumption that communication is one-way, from node to fusion center.

We also leverage the inherent flexibility of CRN systems, assuming a network of identical, frequency agile cognitive radar nodes which are able to choose non-overlapping frequency bands and waveforms from a finite library. 
This agility allows each independent cognitive radar node in the network to avoid other emitters in the environment. 
In addition, the described CRN uses sensing and agility to avoid causing harmful interference within the network. 
We discuss the use of a library of orthogonal waveforms, which has been shown to result in both lower ambiguity\footnote{This holds even though the bandwidth of the network is divided among the nodes instead of being allocated to a single node. } \cite{681325} and reduced mutual interference \cite{Majumder2014}.

Our approach applies iterative, \emph{online} Reinforcement Learning (RL), enabling a network of independent identical cognitive radar nodes (or just ``nodes'') to collectively optimize their frequency band and waveform selections over time. 
Each node in the network has the same goal: to optimize target tracking and detection for the entire network. 
This means that actions which cause low Signal to Interference plus Noise Ratios (\texttt{SINR}) for any node in the network are penalized. 
One of the ways this is accomplished is by detecting instances of ``collisions'', or mutual interference, which occur when more than one node selects the same waveform and center frequency in the same PRI.

This work described a method for using two simultaneous techniques to enable a cognitive radar node to select both a center frequency (i.e. frequency band) and a waveform in each time step. 
First to select a center frequency, a MMAB algorithm is implemented which observes rewards from the environment, and collisions between radar nodes. 
The goal of this first algorithm is to obtain an \emph{optimal frequency allocation}, which is defined later, but is an assignment which maximizes tracking performance. 
Second, in order to select waveforms in each time step, an independent single-player bandit model is instantiated in each frequency band, for each radar node. 
This enables the learner to focus on interactions within the band. 
The single-player bandit algorithm is only concerned with rewards from the environment, since over time the frequency allocation will become free of mutual interference.

Multi-armed bandit algorithms were introduced in the 1950s \cite{bams/1183517370} to study the sequential decision-making problem. 
Specifically, the authors considered the problem of maximizing expected rewards when drawing from two or more differently distributed random sequences. 
The problem of \emph{multiple players} interacting in this environment was not addressed in the literature until 2010 \cite{5535151}. 
The multi-player problem was initially motivated by, among other things, coexistence in cognitive radio applications. 
We first adapted MMAB algorithms to the cognitive \emph{radar} application in \cite{howard2021multiplayer}. 

\subsection{Problem Summary}
This work considers a problem of sequential decision making. 
In each of many time steps, how should a group of distributed radar nodes select a waveform and center frequency?  
The use of \emph{multi-player multi-armed bandit} models is discussed to address this problem. 
We will also investigate and mitigate the effect of non-cooperative transmitters (primary users) which can cause unacceptable levels of interference. 
Such a network must use some algorithm to select frequencies and waveforms. 
Any pre-allocation strategy places assumptions on the interference, which is not known. 
In addition, pre-allocating actions to avoid collisions will not necessarily maximize utility, as we will describe.

Given this, our goal is to describe the two intertwined algorithms for band and waveform selection which a distributed radar network could use to self-organize. 
The presence of a fusion center is still assumed, which is capable of collecting observations from each node in each CPI for the purpose of target tracking. 
However, this fusion center provides no support in the network's efforts to determine which actions each node should take, due to the communication and coordination issues discussed above. 
Measurements are sent to a fusion center and combined on the millisecond timescale, while decisions are made on microsecond intervals. 

This work considers two types of interference. The first, which is referred to as \emph{mutual interference}, is when one radar node interferes with another. 
As is shown later, this can have a detrimental effect on target tracking and detection due to unacceptable Signal to Interference plus Noise Ratio (\texttt{SINR}) levels. 
The second, which can be termed \emph{outside interference}, is when another system causes interference to one or more radar nodes. 
Both of these are to be avoided.

We model a network of pulsed radars, where once per pulse repetition interval (PRI) each radar in the network will select a waveform and transmit it on a chosen carrier frequency. 
As we will describe later, the nodes must share a fixed total bandwidth of $B$, which is not dependent on the number of nodes.  
Each radar node must estimate when it encounters interference from other nodes, using a method we will describe. 
Each radar is capable of steering the main beam, but due to sidelobe levels and environmental scatter, some energy from each transmission is received at all nodes in the network.
If any two nodes select interfering actions, they will experience degraded performance which we will demonstrate.

\subsection{Contributions} 
This paper extends earlier work \cite{howard2021multiplayer}, \cite{howard2022adversarial}. 
We previously investigated the available MMAB algorithms, and in \cite{howard2021multiplayer} studied reward structures which differ between radar nodes. 
In \cite{howard2022adversarial}, we analyzed algorithms tailored towards \emph{adversarial} environments, where the environment has knowledge of the algorithm being used by the radar network, and can pre-select a reward sequence in an attempt to worsen the performance of the network. 
Rather than focus on these different environment classes, this current work is interested in the following. 
\begin{itemize}
    \item Developing the system model for the MMAB problem in CRNs. We provide the first such model for MMAB algorithms in this context. 
    \item Providing a novel two-level algorithm capable of converging towards optimal center frequency and waveform selection for radar operation. Our proposed online learning technique is capable of converging to an optimal solution under sub-linear cumulative regret \cite{besson2018multi}, which corresponds to radar tracking estimation improving over time. 
    We provide discussion on this self-organizing CRN can avoid both mutual and non-cooperative interference. 
    \item Analyzing the supporting mathematics for a MMAB algorithm used for waveform selection in cognitive radar networks. We discuss the reward scenarios and decision making structure necessary to apply the MMAB techniques to the CRN problem. 
    \item Demonstrating our proposed technique against alternatives such as SAA and fixed allocation in simulation. We then provide conclusions based on our simulations. 
\end{itemize}

Combined with the two earlier publications, this paper represents the first work on the topic of MMAB algorithms for coexistence in cognitive radar networks. 

\subsection{Notation} We use the following notation. 
Matrices and vectors are denoted as bold upper $\mathbf{X}$ or lower $\mathbf{x}$ case letters.
Functions are shown as plain letters $F$ or $f$. 
Sets $\mathcal{A}$ are shown as script letters. 
The cardinality $|\mathcal{A}|$ of a set $\mathcal{A}$ refers to the number of elements in that set. 
The logical negation of a statement $a$ is given by an overline $\overline{a}$. 
The transpose operation is $\mathbf{X}^T$. 
The backslash $\mathcal{A}\backslash \mathcal{B}$ represents the set difference as $\mathcal{A}\backslash \mathcal{B} = \{a\in\mathcal{A} : a \not\in \mathcal{B}\}$. 
Indicator functions of a variable $x$ on a set $\mathcal{A}$ are denoted as $\mathbbm{1}_\mathcal{A}(x)$. 
The set of all real numbers is $\mathbb{R}$ and the set of integers is $\mathbb{Z}$. 
The speed of electromagnetic radiation in a vacuum is given as $c$. 
The imaginary number is $i$. 
We use $\mathcal{U}\{\mathcal{A}\}$ to denote the uniform distribution over a space $\mathcal{A}$. 


\subsection{Organization} The rest of this paper is organized as follows. 
In Section II, we review previous work on radar networks and cognitive systems.  
Section III develops our proposed machine learning techniques to assign center frequencies and waveforms to each radar in a network. 
We first detail a method for frequency band selection, followed by a method for waveform selection within that frequency band. 
Section IV provides simulation results and discussion, and in Section V we draw conclusions and suggest future work.  

\section{Background}

\subsection{Cognitive Radar Networks}
Cognitive radar, first described by Haykin \cite{haykin2006}, is described as having four main elements: the Perception-Action Cycle (PAC); memory; attention; and intelligence \cite{6218166}, \cite{7910111}. 
The perception-action cycle describes the iterative interactions typical to this sort of system. 
Further, the IEEE has defined cognitive radar as systems which display intelligence and can modify both operating and processing parameters in response to a changing environment \cite{8048479}. 
This could also refer to spectrum sensing, where a radar would observe spectral behavior, and adapt based on what it sees.

Cognitive radar systems have been proposed for non-cooperative coexistence many times \cite{8000669, 9532494, 8961364}. 
Specifically of interest in this work are cognitive radar \emph{networks}, which can combine observations from individual radars in an intelligent manner to benefit the network \cite{1574168, 6586147}.

Cognition in radar networks can be accomplished in two main ``modes'': 
\begin{itemize}
    \item \textbf{Distributed cognition}, where each node in a network possesses cognitive ability, and
    \item \textbf{Centralized cognition} where a central coordinator makes network-wide decisions. 
\end{itemize}
This poses a fundamental trade-off: centralized cognition may appear to offer several benefits, most significantly a reduced problem space\footnote{In a centralized scenario, a coordinator can limit the choice of actions to the space which contains no collisions, which we will define later. However in a decentralized scenario, the problem space expands to include conflicting actions. } due to the combined observations, but it has significant drawbacks. 
Any centralized system becomes vulnerable to the loss of the central node. 
In addition, reporting observations on a realistic time scale may require communications at the PRI timescale while sharing target information which requires, at most, communications at the coherent pulse interval (CPI) timescale (typically two orders of magnitude slower).   
Since spectrum access for traditional monostatic radar is already limited, and an increased number of nodes will worsen the problem, cognitive radar networks will benefit from any reduction in wireless communication overhead.

In the context of radar, cognitive systems have been described using the PAC \cite{9532494, 9114775, 9527128}, a closed-loop model which implements cognition as defined above. 
It has been stressed that a major property of this problem is the \emph{interaction} between the cognitive system and its environment \cite{8961364}. 
To that end, we can visualize the cycle as shown in Fig. \ref{fig:block_diagram}, where each node uses a cognitive process (multi-player bandit frequency selection paired with a ``single-player'' waveform selection algorithm in this work) to choose a waveform and center frequency from a library to transmit in each time step.

We use the following notation to describe the network structure: 
\begin{itemize}
    \item \textbf{Centralized network}: A CRN using a \emph{central coordinator} to assign frequencies and waveforms, as well as to fuse network observations. The central coordinator has bidirectional communication with all of the nodes. 
    \item \textbf{Decentralized network}: A CRN with \emph{no feedback} with a \emph{fusion center} to combine network observations. Information only flows from the nodes to the fusion center; there is no feedback. 
\end{itemize}
\begin{figure*}
    \centering
    \includegraphics[scale=1.3]{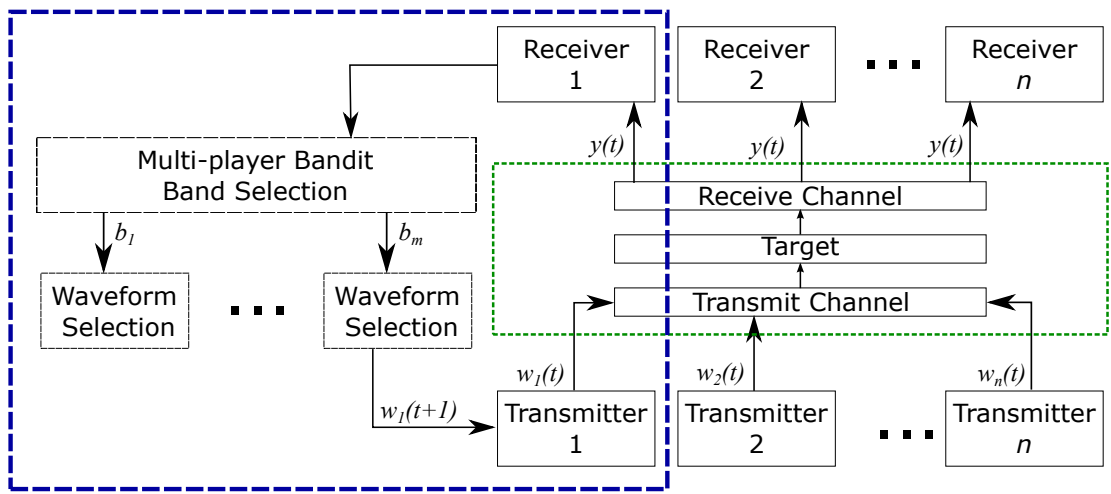}
    \caption{Transmit/receive cycle for the cognitive radar network. The decision process for the first transmitter/receiver pair has been shown, but is implemented at each node. Importantly, each node $i$ independently selects a waveform ${w}_i(t)$, which is modulated by the environment, then returned as a waveform $y(t)$. Using the received energy, the cognitive learner selects the next transmit waveform ${w}_i(t+1)$. }
    \label{fig:block_diagram}
\end{figure*}

\subsection{Related Work}
Radar networks have several methods to track and detect targets. 
Generally, radar networks fall into two classes \cite{MIMObook}: MIMO (multi-input, multi-output) radar networks and radar sensor networks (RSNs). 
We focus primarily on the latter, since they are well-suited to distributed and decentralized techniques. 
RSNs utilize multiple radars that are independent of one another and operate in a mono-static mode. 
Each radar makes independent observations of the environment and forms target models.
The central node, if present, can provide feedback over a longer time period and combine sensing observations from all the radars.

A fundamental work on RSN's is \cite{rsndesign}, where the authors propose a basic framework for RSN target detection. 
In their work, it is assumed that each radar conducts observations, then losslessly communicates the received waveforms to a clusterhead for signal processing. 
They show that using Continuous Wave waveforms which do not overlap in frequency, coupled with the spatial diversity inherent to RSN's, provides a significant improvement compared to the baseline single monostatic sensor. 
Since the waveforms do not overlap in frequency we can consider them to be orthogonal under the definition provided later.

An early work on CRNs was published by Haykin shortly after his seminal work on cognitive radar \cite{haykin2006}. 
In this work, Haykin describes potential future applications for CRNs including the use of centralized cognition for data fusion and control. 
Another implementation, with more emphasis on distributed cognition, is described in \cite{mimogame}. 
A MIMO CRN is split into clusters, with no assumption of communication between clusters. 
A noncooperative game is used to formulate a power allocation algorithm between clusters to minimize mutual interference. 

In this work we assume a decentralized CRN with a \emph{fusion center} that only performs the data fusion function, i.e. it provides no feedback to the nodes. 
To fuse target information, each node reports target observations on some schedule, i.e. once per PRI or CPI. 
To avoid solving the implied communications problem, we assume that this unidirectional communication happens once per CPI, occurs on some pre-allocated communication channel separate from radar frequencies, and is error free.

Machine learning techniques are well suited to this problem due to the variable nature of Radio Frequency (RF) environments. 
Instead, an approach based on pre-allocation of waveforms and frequencies could be selected. 
This, however, makes strong assumptions on the interference environment, as well as the quality of the target measurement each radar node can make. 
In addition, traditional methods will place more stringent assumptions on the coordination of the network, making the distributed cognition scheme less powerful. 
A cognitive radar \emph{network} is selected for this problem specifically for this reason: while a single node may be overwhelmed by an interferer or lack proximity or high Radar Cross Section (RCS), the inherent spatial diversity of a CRN offers flexibility in many domains, ensuring better tracking accuracy and higher probability of detection \cite{spatialdiversity}.

In previous RL-based DSA work, the Markov assumption is often applied \cite{radar_emitter_hiddenmarkov}, \cite{9069470}. 
Markov environments have the property that the next state depends \emph{only} on the current state. 
The goal in this sort of problem, then, is to learn the probability distribution of next states conditioned on the current state. 
However, interferer behavior in realistic channels often has high temporal correlation, leading to the breakdown of the Markov assumption. 
This motivates study into online sequential learning, which seeks to exploit any patterning or correlation in previous observations to inform future decision-making.

Mutual interference mitigation is a common goal for machine learning systems in radar sensor networks. 
A common target application is automotive radar in connected vehicles, where multiple independent radars attempt to share a dynamic environment with many radars entering and leaving a region. 
This obviously causes problems for any technique relying on pre-allocated frequencies. 
It is important to note that this application has results that are very easily generalized to radar sensor networks. 
The work of \cite{8828037} provides a good overview of the problem, current mitigation techniques, and possible future directions. 
One technique discussed is beamforming to steer nulls in the directions of other nodes in the network. 
This can be accomplished through estimating the location of the other nodes, then selecting waveforms to null specific directions. 
This technique limits the target tracking and detection capability of any network, since energy is necessarily not directed along the Delauny lines (i.e., those lines connecting two nodes) of a network. 
Any target in these regions would have a lower probability of detection due to the relatively lower power emitted in that direction.

Another common method discussed in \cite{8828037} is to borrow Code Division Multiple Access (CDMA) waveform techniques from wireless communications. CDMA waveforms have the benefit of being orthogonal in code, which allows for separability on receive even with overlaps in other dimensions. 
CDMA commonly has the downside of requiring a broad bandwidth, but this issue is somewhat mitigated in the radar application since wideband channels are usually available. 
However, a somewhat larger issue is the requirement for synchronization in CDMA schemes which is difficult to guarantee. 

In addition, the huge power disparities caused by variations in target and intra-node distances tend to cause swamped target returns, causing large problems in filtering. 
Proposed future work includes joint radar/communication systems, decentralized multiple access schemes, and use of future modulation techniques to mitigate interference.

In \cite{8388252} the authors provide analytical bounds for network wide Signal to Interference Ratio (SIR) levels due to different waveform selections. 
The author of \cite{4106078} shows probabilities of mutual interference, considering spatial, temporal, and frequency domains. 
Importantly, they conclude that regardless of temporal delay or spatial separation, mutual interference due to frequency overlap tends to dominate. 
Thus, in this work, we seek to avoid such interference by choosing orthogonal frequency bands.

To consider the case of non-orthogonal waveforms, \cite{intsup} presents a cross-matched filtering technique to mitigate mutual interference between radars. 
Non-orthogonal waveforms used across a RSN can cause high-power backscatter \& line of sight interference which can severely impact performance. 
The cross-matched filter is motivated by the fact that perfectly orthogonal transmitted signals may not remain orthogonal upon receive, so there is a need to sequentially match filters to each transmitted waveform in the network to suppress interference, which leaves only the transmitted waveform. 
This is predicated on the assumption that perfect orthogonality is needed for coherent processing.

In \cite{Majumder2014}, the authors show that near-orthogonality is acceptable, and introduce the idea of $\varepsilon$-orthogonality to provide a bound on the similarity between waveforms. 
$\varepsilon$-orthogonality is defined based on the cross-ambiguity between two waveforms, which is related to cross-correlation. 
However, the authors do not consider the case when large power disparities exist, such as between radar reflections and line of sight interference.

Cross-ambiguity is a time ($\tau$) and frequency ($\nu$) analysis tool used to solve signal processing problems. 
Specifically, we will look at cross-ambiguity as a way to reduce mutual interference in RSN's. 
A thorough mathematical description of the cross-ambiguity function is provided in \cite{Vandenberg2012}. 

Mutual interference detection and mitigation in radar systems has been addressed from several directions. 
\cite{8461806} describes several coding schemes for automotive frequency-modulated continuous wave radars which eliminate the mutual interference problem without requiring detection. 
For our scenario, which uses pulsed radar waveforms, this solution is incompatible. 
Similarly for automotive radar, \cite{7226239} describes an energy detection method for mutual interference detection, enabling a higher-level process to determine a mitigation scheme. 
Once detected, the authors propose the use of spatial filters to mitigate the effects of mutual interference. 
Unfortunately, one of the limitations of this algorithm is that if the target happened to be in the same direction as an interfering node, this method would result in failed detections and poor tracking.

To summarize, one clearly important aspect of resource allocation in radar sensor networks is the consideration of multiple access schemes at each node. 
By selecting resources correctly at each node, mutual interference can be greatly avoided. 
This problem draws a lot from the similar problem of interference mitigation in wireless communication networks, but with an important distinction: wireless communication aims to maximize the information transmitted directly to some receiver, while radar networks seek to maximize target information using co-located transmitters and receivers. 
In this work we consider multiple access primarily through frequency and code division, due to the limitations of a radar network. 
We will use a library of nearly orthogonal waveforms and orthogonal frequency bands to detail a method for decentralized waveform selection at each node, requiring minimal coordination. 
\section{Optimal Waveform Assignment}
Here we will describe the space of possible outcomes, when each cognitive radar node in a network must select both a \emph{center frequency}\footnote{We use both terms \emph{center frequency} and \emph{frequency band} interchangeably to refer to the bandwidth allocated to a channel. } and a \emph{waveform}\footnote{The term \emph{waveform} refers to the specific waveform selected from a library to be transmitted in a particular frequency band. }. 
Specifically, the network seeks to select the actions resulting in the highest \emph{utility}, when there is no communication between the independent nodes.

\subsection{Center Frequency Selection}
We can begin by defining an optimal configuration based on a pre-defined set of frequency bands.
We assume that the center frequencies are selected such that two radars on two different center frequencies will have non-overlapping bandwidths. 
In each of many time steps $t$ less than finite time horizon $T$, each radar $r_i$ selects a center frequency $f_j$. 
Let the set of radars be denoted as $\mathcal{P} \defeq \{r_1, r_2, \dots, r_i, \dots, r_M\}$ with $|\mathcal{P}| = M$. 
Similarly, let the library of center frequencies be $\mathcal{F} \defeq \{f_1, f_2, \dots, f_j, \dots, f_N\}$ with $|\mathcal{F}| = N$.
Let $M$ and $N$ be related as  $M,N \in \mathbb{Z}$ s.t. $M < N$. Note that $\mathcal{P}$ and $\mathcal{F}$ do not have a time dependence.

Then, $\mathcal{P}$ and $\mathcal{F}$ form the two disjoint sets of nodes of a fully connected bipartite graph.  
The matrix of choices\footnote{By ``matrix of choices'' we mean the space of possible radar-frequency pairings. } becomes the set of edges $\mathcal{E} = \mathcal{P} \times \mathcal{F}$. 
We specify a fully connected graph to imply that any radar $r_i \in \mathcal{P}$ can select (and observe a reward for selecting) any center frequency $f_j \in \mathcal{F}$. 
Further we can call the graph $(\mathcal{P}, \mathcal{F}, \mathcal{E})$.
Formalizing the problem in this manner allows for the examination of the space of possible outcomes, when one node selects one frequency band. 
We can analyze the utility of each possible configuration, to get a sense of the objective.

An edge is any connection between the vertices $r_i, f_j$ of the bipartite graph. 
In this application, edges are equivalent to a radar $r_i$ selecting center frequency $f_j$. 
Edges can be denoted as the concatenation of the vertices they connect - $r_if_j$. 
We can assign each edge $r_if_j$ a weight $u(r_i,f_j) \defeq \mu_{i,j}$ with 
\begin{equation}
    \label{eq:reward_map}
    u(r_i,f_j):\mathcal{P}\times \mathcal{F} \to [0,1] 
\end{equation}
which is the reward that radar $r_i$ observes for selecting center frequency $f_j$. 

A mapping $\widehat{\pi}(t)$ at time $t$ is any collection of edges where no vertex in $\mathcal{P}$ is repeated. 
Mappings may or may not contain repetitions of vertices in $\mathcal{F}$. 
In other words, $\widehat{\pi}(t)$ can be any function which maps $\mathcal{P}$ to $\mathcal{F}$. 
If there exist $a,b\in\mathbb{Z}$ such that $r_a$ and $r_b$ are in the same mapping ($r_af_j\in \widehat{\pi}(t)$ and $r_bf_j \in \widehat{\pi}(t)$), then we say that radars $r_a$ and $r_b$ have collided. 
For any time step $t$ let $\mathcal{C}_{\widehat{\pi}(t)}$ be the set of colliding radars for a mapping $\widehat{\pi}(t)$.
Formally, 
\begin{definition}[Collision]Let $\mathcal{C}_{\widehat{\pi}(t)}$ be the set of colliding radars: 
$\mathcal{C}_{\widehat{\pi}(t)} \defeq \{r_i \;|\; \exists \; r_i, r_k, \;$\emph{with}$\; r_if_j, r_kf_j \in \widehat{\pi}(t)\}$. \\
A radar $r_i$ is said to \emph{collide} at a time step $t$ if \emph{$r_i \in \mathcal{C}_{\widehat{\pi}(t)}$}.
\end{definition}

A matching $\pi(t):\mathcal{P} \to \mathcal{F}$ is any mapping with no common vertices\footnote{In other words, matching functions are injective. }. 
So, in a matching, there are no collisions. 
We can collect all of the possible matchings $\{\pi\}$ of a graph $\left(\mathcal{P}, \mathcal{F}, \mathcal{E}\right)$ into a set $\mathcal{M}$. 
Note that if $\pi(t)$ is a matching, it is also a mapping, while if $\widehat{\pi}(t)$ is a mapping, it is not necessarily a matching. 
\begin{remark}
If $\pi(t)$ is a matching, the set of colliding radars $\mathcal{C}_{\pi(t)}$ is empty. 
\end{remark}
Note that we'll sometimes drop the time dependence of a mapping when we discuss the space of possible mappings. We can determine which matchings are best by looking at the sum of their combined rewards. 
We call this value the \emph{utility} of the matching and define it as follows. 
\begin{definition}[Utility]
The \emph{utility} of a mapping $\pi$ is the sum of the rewards each radar observes using that mapping: 
\begin{equation}
    \label{eq:matching_utility}
    U(\pi) = \sum_{r_i f_j \in \pi} \mu_{i,j}
\end{equation}
\end{definition}
Now, adding time dependence, we can find the maximum utility in a step $t$ as 
\begin{equation}
    \label{eq:max_utility}
    U^*(t) = \max_{\pi \in \mathcal{M}} U(\pi(t))
\end{equation}

This leads naturally to a definition of an optimal configuration. 
This represents the highest-reward center frequency selection for each radar in a network for a given time step $t$. 
\begin{definition}[Optimal Configuration]
A set of radar nodes $\mathcal{P}$ with a common center frequency library $\mathcal{F}$ is said to be in an \emph{optimal configuration} at time step $t$ if $\pi(t):\mathcal{P} \to \mathcal{F}, \pi(t) \in \mathcal{M}$ has utility $U^*(t)$. 
\end{definition}

\begin{remark}
Existence of an optimal configuration is guaranteed by definition, but uniqueness of the solution is not. 
\end{remark}


\subsection{Orthogonal Waveforms}
So far, we have defined optimality based on selection from a arbitrary waveform library. 
We will make this more concrete by discussing the use of \emph{orthogonal waveforms} to develop a library, and justify the resulting \texttt{SINR} differences. 
Specifically we will refer to the \texttt{SINR} post matched filtering, so that orthogonal interference does not impact the observations of each radar node. 

The use of orthogonal waveforms in radar networks has been shown to aid in target tracking by obtaining more information through the differences in matched filter ambiguity responses \cite{681325}. 
Each node might have a different view of a target, local interference, or local scattering. 
Each of these could lead to differing local preferences for waveforms or frequency bands. 
This problem is somewhat simpler in the centralized setting, when a coordinator can specify a matching of orthogonal waveforms in each time step for the network. 
One limitation of the centralized case is that the central controller may not have sufficient information regarding local conditions. 
However, in this work we examine the decentralized setting. 
Since we have no communication within the network on the PRI time scale, there will be no means for the network to assign a different orthogonal waveform to each node for each PRI.

Two waveforms are defined as orthogonal when their time cross-correlation is zero. 
In radar applications, due to target motion and noise, waveforms that are orthogonal at transmission are not necessarily orthogonal at the receiver. 
Therefore it has been proposed to define \textit{near}- or $\varepsilon$-orthogonality.
Two waveforms are nearly orthogonal when their cross-ambiguity is less than $\varepsilon$ \cite{Majumder2014}. Formally, 
\begin{equation}
    \label{eq:nearly_orthogonal}
    \max_{\tau, f_s}\frac{|\chi_{w_1,w_2}(\tau,f_s)|^2}{E_{w_1}E_{w_2}} \leq \varepsilon
\end{equation}
where $\chi_{w_1,w_2}$ is the cross ambiguity
\begin{equation}
    \label{eq:cross_ambiguity}
    \chi_{w_1,w_2}(\tau,f_0) = \int_{-\infty}^\infty w_1(t)w_2^*(t-\tau)e^{i2\pi f_0 t}dt
\end{equation}
and $E_{w_i}$ is the energy of waveform $w_i$. 
This can be calculated for either the transmitted or received energy of waveform $w_i$.

Intuitively, very different signals will have low cross-ambiguity. 
This holds for both temporally shifted (TDMA) and frequency shifted (FDMA) signals, as well as signals that are coded to be orthogonal (CDMA). 
However, we will not consider TDMA for the following reasons.

Consider a network using TDMA waveforms with equal PRFs. 
Then, every PRI, there would be a chance of receiving mutual interference at one or more time delays. 
Conceptually, the duration of a PRI could be increased to $M \times$PRI to allow TDMA where there are $M$ radar nodes. 
This would be very inefficient and also result in limited maximum detectable range. 
If any of these align with target returns, then the target will not be resolvable.

Consider next a network using TDMA waveforms with coprime PRFs. 
Then, each radar would be unable to shift its PRF with high fidelity to optimize target tracking. 
This approach also assumes some sense of PRF coordination network-wide. 
Since PRF should be linked to target radial velocity, this places a large constraint on target tracking ability.

Due to these reasons we will not consider TDMA for waveform selection. 
Instead, we will assume a library of Orthogonal Frequency Division Multiplexing (OFDM) waveforms in each frequency band.
We specify two main design constraints on this library: 
\begin{enumerate}
    \item Mutual collisions must remain detectable (under criteria described later). 
    \item The library must contain a waveform which is orthogonal to a single narrow-band interferer in one of $2^s$ sub-bands\footnote{The term sub-band refers to subdivisions of the selected frequency band. }. 
\end{enumerate}
We make the first constraint due to MMAB algorithms requiring collision information: if more than one node selects the same frequency band, they must detect it. 
Secondly, we model each frequency band as containing at most one narrow-band primary user such as a communications system. 
This information is not known a priori, so part of the learning problem is understanding where the primary user is since a waveform cannot be orthogonal to interference that is unknown in advance. 
We will detail a method for this in a following section.

Assume a set $\mathcal{H}$ of $s$ sub-bands, enumerated $h_1, h_2, \dots, h_s$ with $s$ even. 
Further, let each sub-band $h_i$ be contiguous with $h_{i+1}$.
Let waveforms $w_j$, $j\leq s + 1$ occupy sub-bands $H_j \subseteq \mathcal{H}$, where 
\begin{equation}
    H_j = \begin{cases}
    \text{maxB}\left(\mathcal{H}\backslash h_j\right), & j < s+1\\
    \mathcal{H}, & j = s+1
    \end{cases}
\end{equation}
The backslash $\mathcal{X}\backslash x$ denotes set subtraction. 
Let $\text{maxB}\left(x\right)$ be the function which returns the largest-bandwidth contiguous set of sub-bands.

As an example, when $s=4$, $\mathcal{H} = \{h_1, h_2, h_3, h_4\}$. Waveform $w_1$ will use sub-bands $H_1 = \{h_2,h_3,h_4\}$ and $w_2$ will use sub-bands $H_2 = \{h_3,h_4\}$. 
This method provides a set of waveforms which can be orthogonal to an interferer in any given sub-band, without requiring notched waveforms.

Note that if two radar nodes were to use nearly orthogonal waveforms, \emph{there may still be unacceptable levels of interference} due to the low power echo received from the target, relative to the direct signal from the second radar (even with low sidelobes). 
In addition, any non-synchronous reception will destroy the orthogonality. 



%
%
%

So far we have had no discussion of collision detection. 
How are radar nodes able to sense when a collision has occurred? 
We will examine this in the following section.

\subsection{Collision Detection}
We would like for each node $i$, which picks waveform $w_j$ in PRI $p$ to form an estimate of $\mathbbm{1}_{\mathcal{C}_{\mathcal{W}_p}}(w_j)$, where 
\begin{equation}
    \label{eq:ind_col}
    \mathbbm{1}_{\mathcal{C}_{\mathcal{W}_p}}(w_j) = \begin{cases}
    1, & w_i \in \mathcal{C}_{\mathcal{W}_p}\\
    0, & \text{else}
    \end{cases}
\end{equation}
with $C_{\mathcal{W}_p}$ being the set of colliding waveforms. 
So, let $I$ be an estimate as
\begin{equation}
    \label{eq:est_col}
    I = \widehat{\mathbbm{1}}_{\mathcal{C}_{\mathcal{W}_p}}(w_j). 
\end{equation}
%


Recall the radar equation: 
\begin{equation}
    \label{eq:radar}
    P_r = \frac{P_t G_t G_r \lambda^2 \sigma}{(4\pi)^3 R_t^2 R_r^2}
\end{equation}
where $P_t$ is the transmission power, $G_t$ is the transmitting array gain, $G_r$ is the receiving array gain, $\lambda$ is the wavelength corresponding to a center frequency of $f_c$, $\sigma$ is the radar cross section (RCS) of the target, $R_t$ is the distance from emitter to target, and $R_r$ is the distance from target to receiver. 
Note that in a mono-static scenario such as the one we consider, $G_t=G_r$ and $R_t=R_r$ for a given radar node. 
We can compare the radar equation, which describes the power received from a target, to the power received from an interfering radar (or other) transmission at a range of $R_n$: 
\begin{equation}
    \label{eq:int_power}
    P_I = \frac{P_tG_t(\theta)G_r \lambda^2}{4\pi R_n^2}
\end{equation}
where $G_t(\theta)$ represents the interfering radar transmit gain in the direction of the radar of interest. 
We can compare the received power from both types of source, assuming the same distance and gain characteristics, and an RCS of $3m^2$, which is characteristic of a medium aircraft \cite{knott2006radar}. We assume for this example that the radar node operates with an output power of 30dBw. 
Specifically, we will compare three levels of received power. First, the power received from a main-beam target reflection. Second, from the sidelobe (at a relative -10dB from the main beam) of a second radar node, and third from an outside interferer with an output power of 10dBw. 
In Fig. \ref{fig:rx_power}, we can see that a radar node will receive much higher power from another node than it would from returns of its transmitted waveform or from interferers. 
\begin{figure}[t!]
    \centering
    \includegraphics{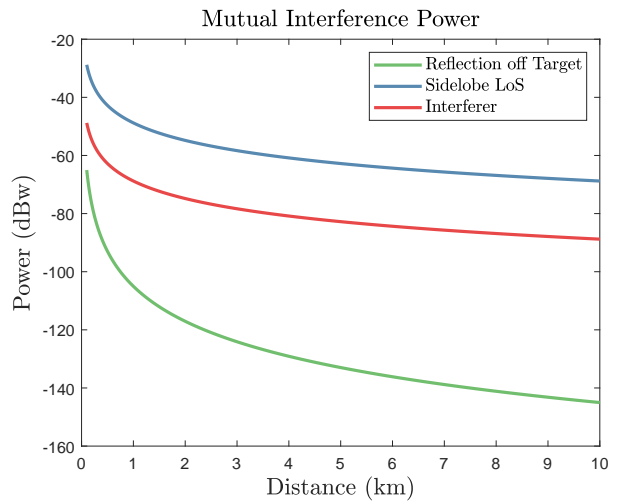}
    \caption{Received power over a range of distances for a target, sidelobe LoS received from another node in a network, and LoS interference. }
    \label{fig:rx_power}
\end{figure}

The lowest-power mutual interference is still more powerful than most radar returns. 
If we consider main-beam interference instead, the effect would be even more evident. 
So, we propose the use of a received power threshold to determine when collisions occur, given some small limitations.  
The remainder of this paper considers distances between $1$ and $3$km. 

If in a given PRI, an individual radar receives a return composed of a target reflection along the main beam AND high-power mutual interference at an angle $\theta$, we can represent the received power as: 
%
\begin{equation}
    \label{eq:combo_power}
    P_{r,I} = \frac{P_t G^2 \lambda^2 \sigma}{(4\pi)^3 R_t^4} + \frac{P_t GG(\theta) \lambda^2}{(4\pi)^2 R_n^2}
\end{equation}
if the mutual interference is totally orthogonal to the radar return. 
Then, assuming that $\frac{\sigma}{4\pi R_t^4} < \frac{1}{R_n^2}$, we can see that
\begin{equation}
    \frac{P_t \lambda^2 GG(\theta)}{(4\pi)^2 R_t^2} \leq P_{r,I} \leq \frac{2 P_t \lambda^2 GG(\theta)}{(4\pi)^2 R_t^2}
\end{equation}
Note that this holds even when the mutual interference and radar return are not completely orthogonal due to the inequalities. Using this as an energy detector, we are interested in PRIs when the received energy exceeds 
\begin{equation}
    \int_{t}^{t + PRI}\frac{P_t G^2\ \lambda^2}{(4\pi)^2 R_n^2} dt
\end{equation}
where a PRI lasts from $t$ to $t+PRI$.

Then, two issues remain before we can define a collision detector. 
First, how do we determine the value of $R_n$, the distance to the nearest node? 
In the scenario where each radar node is aware of the network positions, this is trivial. 
However, in the general case, when the network is simply distributed as a point process with some density, we can determine the average distance to the nearest node.

Specifically let the radar network be spatially distributed according to a Binomial Point Process (BPP) with intensity $\gamma$. 
BPPs are spatial distributions with a fixed number of points $N(t) = n$. 
Poisson Point Processes (PPPs), on the other hand, are spatial distributions with $N(t)$ drawn from a Poisson distribution. 
A BPP is equivalent to a PPP conditioned on $N(t)$. 
Then, the nearest neighbor distribution function is given as $g(r)$ \cite{haenggi_2012}. 
Let the inter-node distance $R_n = \mathbb{E}[g(r)]$ be the expectation of the nearest neighbor distribution. 

Secondly, we need to support the assumption that $\frac{\sigma}{4\pi R_t^4} < \frac{1}{R_n^2}$. 
In Fig. \ref{fig:distance_ratio} we can see that for a range of RCS values and average node spacings, the minimum target distance needed for this assumption to hold is consistently low, and is sublinear with average node spacing. 

\begin{figure}
    \centering
    \includegraphics{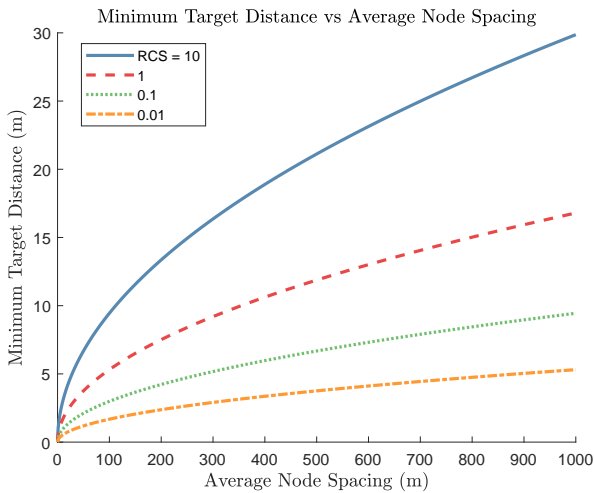}
    \caption{A comparison of the distances needed for the collision detection assumption to hold, given various RCS values. }
    \label{fig:distance_ratio}
\end{figure}

We can write the final energy detector as 
\begin{equation}
    \label{eq:col_estimator}
    I(t) = \begin{cases}
    1, & E_I(t) \geq \int_{t}^{t^*}\frac{P_t G^2\ \lambda^2}{(4\pi)^2 \mathbb{E}[g(r)]} dt\\
    0, & \textit{else}
    \end{cases}
\end{equation}

Note that we discuss collisions only in terms of \emph{mutual interference}. 
We consider instances of outside interference to primarily impact the observed reward, and do not consider these as collisions. 
Since outside interferers are assumed to transmit at substantially lower power than each radar node, this assumption is justified.

\subsection{Band and Waveform Selection}
We can extend the decision-making to two discrete levels by parameterizing the waveforms by an index of $\mathcal{W}$, the library of OFDM waveforms, and an index of $\mathcal{F}$ which we define to be the set of available center frequencies. 
Then, the action comes from the space $\mathcal{F}\times \mathcal{W}$. 
Specifically we will develop the use of a multi-player bandit algorithm for center frequency selection, then delegate the choice of a waveform to a single player bandit, with independent instances for each choice of center frequency.


There are two primary motivations for this dual structure. 
First, the outer multi-player bandit algorithm will be able to learn the average behavior of each sub-band, while simultaneously learning to select center frequencies from an optimal matching. 
Second, the inner single-player bandit algorithms have the ability to tailor the waveform to avoid intra-band interference, causing performance to improve over time.


We will assume a fixed maximum bandwidth for all radars, with a variable center frequency. 
Each radar chooses an action from the space $\mathcal{F}\times \mathcal{W}$, with $\mathcal{F}$ chosen so that there is no overlap in frequency when two different center frequencies are chosen.  
In other words, it selects a waveform from the library $\mathcal{W}$ and a center frequency from $\mathcal{F}$ to transmit that waveform.

\subsection{Algorithms}
Varieties of the proposed sequential decision-making problem has been studied in cognitive radio, as well as in cognitive radar \cite{howard2022adversarial}, \cite{thornton2020efficient}.
The approach in this work employs a multi-player multi-arm bandit framework to allow nodes to develop a model of their environment, without relying on a central decision-maker \cite{howard2021multiplayer}. 
This means that each node can only access the information it has observed: rewards for each waveform selected, and an indicator for collisions. 
Since we are assuming a model with no dedicated communication between nodes, and only limited exchanges of information from each node to the fusion center, each node needs to execute independent algorithms that are capable of converging to the optimal utility for the entire network. 
There is abundant work in the literature on the topic of Multi-Arm Bandit (MAB) models for agents which make sequential decisions over a finite time horizon T \cite{MAL-068, bandits, EXP3}.
In addition, this model is a two-layer algorithm as described above: a multi-player algorithm to determine center frequency, and a single-player algorithm to select waveforms. 

\subsubsection{Multi-Player Algorithms}
MABs consider the sequential interaction between a player and an environment, which consists of multiple ``arms''. 
In each of finitely many time steps, the player selects one arm and observes a reward generated by the environment for that arm. 
Over time the player must balance actions which \emph{explore} versus those that \emph{exploit}. 
Exploring actions are those taken to generate a better understanding of the environment. 
Exploitative actions are those that take advantage of prior knowledge of the environment to maximize rewards \cite{MAL-068}. 
Rewards are assumed by each player to be drawn i.i.d. from unknown distributions. 
Over time the player seeks to learn the mean of the distribution for each arm, which we denote as $\widehat{\mu}_f^r$ for player $r$'s estimate of the mean of arm $f$. 
So, a naive strategy might have the player attempt each arm once then select the one with the highest reward for the rest of the game. 
This only works if 1) the arm variance is very low, so the arm ordering can be correctly estimated from only one sample, and 2) the arm means do not vary over time. 
In realistic scenarios neither of these are likely to be true. 
For instance, in this scenario, rewards are influenced by interferer behavior.

Multi-player MAB (or MMAB) models consider the interactions between many players which must use the same set of arms. 
As in the framework described above, this allows two players to select the same action and collide, since
if multiple players see the same ``best'' arm, they will all desire to use it. 
So, it is clear then that MMAB models are well-suited to the problem of distributed radar: Multiple radars must select sequential actions while seeking to minimize target tracking error across a network. 

We choose a structure that allows a MMAB algorithm to select the center frequency $f \in \mathcal{F}$, and a ``single-player'' algorithm to select the waveform $w \in \mathcal{W}$ so that the resulting action is in the space $\mathcal{F}\times\mathcal{W}$. 
To be clear, the goal for each node is twofold: 1) end up in a frequency band with no other nodes and which maximizes \emph{network} \texttt{SINR}, and 2) select waveforms in that frequency band which maximize \texttt{SINR}. 
We will discuss the use of Sense and Avoid for both center frequency and waveform selection. 
This will serve as a baseline for comparison, followed by a discussion of MMAB algorithms such as Musical Chairs for center frequency selection with single-player MAB algorithms such as $\epsilon$-greedy for waveform selection. 
In this structure, the goal for the network is for each radar node to learn to select a center frequency which is different from each other node, as well as waveforms which attain high \texttt{SINR} in the selected band. 
Collisions are observed via the technique described above.

While each radar node only implements a single instance of the center frequency selection algorithm, it will implement a waveform selection algorithm for \emph{each center frequency}. 
This structures the algorithm such that interference near one center frequency does not influence waveforms selected near another center frequency. 

In the following section when we detail a MMAB algorithm we will denote the actions as  $f\in\mathcal{F}$, while if we discuss a single-player algorithm for waveform selection we will use notation $w\in\mathcal{W}_f$. 
If we drop the frequency band dependency on the waveform library it is understood that we're referring to the general case of a waveform selection algorithm instantiated in some frequency band. 
For the particular case of Sense and Avoid in the following section, we denote the actions as waveforms $w\in\mathcal{W}$ but also discuss its use for frequency band selection. 

\begin{remark}
We consider only online-learning approaches, since the described network has no prior knowledge of the environment. Due to short coherence times, variable targets, and possibly changing rewards, each radar must learn the behavior of the environment, which includes rewards, interferer behavior, and target behavior. 
\end{remark}

\paragraph{Sense and Avoid} The simplest algorithm we will consider, called 
``sense and avoid'' 
(Algorithm \ref{algo:SAA}), has each radar select a random action and observe for collisions. 
When collisions happen, they select a new action. Otherwise they keep repeating the previous action. 
Assuming no outside interference, this will result in a reward matching, but it will not necessarily be optimal, since the algorithm does not consider which actions may be better for the radar. 
This allows SAA to use $I(t)$ to determine when to switch actions. 
Note that in this instance we can consider a lower threshold for collisions which includes the event where a pulse overlaps with a primary user. 

\vspace{0.1in}
\begin{algorithm}
\SetAlgoLined
\KwResult{w(t)}
 \eIf{$\overline{I(t)}$}{
  $w(t) = w(t-1)$\;
  }{
  $w(t) = \mathcal{U}\{\mathcal{W}\backslash w(t-1)\}$\; 
 }
 \caption{Sense And Avoid}
 \label{algo:SAA}
\end{algorithm}
\vspace{0.1in}

Recall that $\mathcal{U}\{\cdot\}$ represents uniform sampling over a set, and the backslash represents the set difference. 
In simulation, we will use SAA as a center frequency selection algorithm \emph{as well as} a waveform selection algorithm. 

\paragraph{Musical Chairs}
The ``Musical Chairs'' (MC) algorithm \cite{pmlr-v48-rosenski16}, shown in Algorithm \ref{algo:MC}, is a step up in complexity from SAA. MC develops an estimate of $\mathcal{W}^*$, the set of best center frequencies, by specifying a well-defined exploration period where every player attempts to observe the reward for each action as many times as possible while avoiding collisions. Since there is no coordination the regret incurred during this exploration period is rather high, as collisions will be unavoidable. 
The exploration continues for $T_0$ time steps, where 
\begin{equation*}
    T_0 = \ceil[\big]{\max{\left(\frac{16M}{\epsilon^2}\ln\left(\frac{4M^2}{\delta}\right), \frac{M^2 \log(\frac4\delta)}{0.02}\right)}}
\end{equation*}
is chosen to guarantee a regret bound.  
Here, $M$ represents the number of players, $\epsilon$ is a bound on the correctness of the estimate of $\mathcal{W}^*$, and $\delta$ represents the distance from the $M^{th}$ best reward to the $(M+1)^{th}$ best reward. 

\vspace{0.1in}
\begin{algorithm}
    \SetAlgoLined
    \KwResult{$w(t)$} 
    Input $y_i(t), c_i(t), fixed=\text{False}, t=0$\\
    \vspace{3mm}
        \uIf{$\overline{fixed}$}{
            \uIf{$t \leq T_0$}{
            $w(t) =\mathcal{U}\{\mathcal{F}\}$\;}
            \ElseIf{$t>T_0$}{
                \uIf{I(t-1)}{
                    $w(t) = \mathcal{U}\{\mathcal{W}^*\}$\;}
                \Else{
                $w(t) = w(t-1)$\;
                fixed = true\;}}}
        \Else{
        $w(t) = w(t-1)$\;}
    \caption{Musical Chairs}
    \label{algo:MC}
\end{algorithm}
\vspace{0.1in}

\paragraph{Musical Chairs Top M} The ``Musical Chairs Top M'' (or MCTopM) (Algorithm \ref{algo:MCTopM}) \cite{besson2018multi} algorithm was designed to allow all players time to explore through the use of the Upper Confidence Bound \cite{UCB_fischer} which we represent as $g(t)$. 

\begin{equation}
    g(t) = \widehat{\mu}_f^r(t) + \sqrt{\frac{\log(t)}{2T_w^r(t)}}
\end{equation}

In each round, the player will attempt to pick one of the best $M$ center frequencies (contained in the set $\mathcal{W}^*$), where there are $M$ players. 
If the player successfully selects one of these actions, it is marked as a ``chair'' and the player will select it so long as it remains in $\mathcal{W}^*$. 
If a player collides on action $w\in\mathcal{W}^*$, a new action from $\mathcal{W}^*$ will be selected. 
If an action is marked as a ``chair'' but is no longer in the $\mathcal{W}^*$, the player will use the UCB values to select a new action from $\mathcal{W}^*$. 
One interesting feature of MCTopM is that by virtue of the UCB, the player will tend not to change actions too frequently. 
In a realistic radar context this will be beneficial since there is a cost associated with changing waveforms too often.

\vspace{0.1in}
\begin{algorithm}[t!]
\SetAlgoLined
\KwResult{w(t)}
 \uIf{$\overline{w(t-1) \in \mathcal{W}^*}$}{
  $w(t) = \mathcal{U}\{\mathcal{W}^* \cap \{w_i : g_{w_i(t-1)} \leq g_{w(t-1)}(t-1)\}\}$\;
  }
 \uElseIf{$I(t-1)$ \& $\overline{\text{fixed}}$}{
  $w(t) = \mathcal{U}\{\mathcal{W}^*\}$\; 
 }
 \Else{
  $w(t)$ = $w(t-1)$\;
  fixed = true\;
  }
 \caption{Musical Chairs Top M}
 \label{algo:MCTopM}
\end{algorithm}
\vspace{0.1in}

\subsubsection{Single-Player Algorithms}
Traditional bandit algorithms only consider the behavior of an individual learner, sampling from some set of actions with the goal of minimizing cumulative regret. 
As discussed, we utilize the MMAB algorithm to select a frequency, followed by one single player bandit algorithm \emph{per center frequency} which learns the interferer behavior (and therefore rewards) in that band. 
Many algorithms exist in the literature for this problem \cite{MAL-068}, and we discuss two. 
In addition we can use the above SAA algorithm in the single player bandit role. 

\paragraph{$\epsilon$-Greedy} In addition to SAA, the second algorithm we will consider for waveform selection simply selects a random action with some constant probability $\epsilon$, and selects the highest-average-reward action with probability $1-\epsilon$ \cite{sutton2018reinforcement}. 
This allows the learner to explore different actions with some fixed probability over time, while attempting to ensure low regret by selecting the highest-reward action the rest of the time. 

\paragraph{$\epsilon$-Decaying} As a variation on $\epsilon$-Greedy, $\epsilon$-Decaying \cite{sutton2018reinforcement} selects an $\epsilon$ value as a function of the number of trials. 
As a heuristic we found that setting $\epsilon = \frac{1}{t^{0.8}}$ exhibits the best performance in this setting. 
Fig. \ref{fig:decay_sweep} shows that this value is a local minimum in a range of possible exponents. 
\begin{figure}
    \centering
    \includegraphics{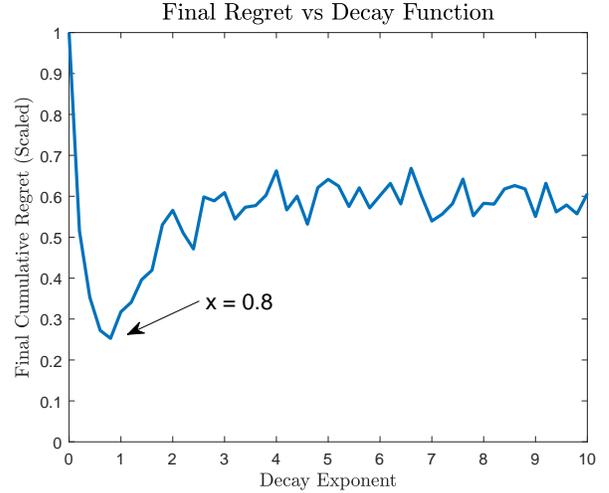}
    \caption{Final regret over a simulation of 100 CPIs with three radar nodes using MCTopM and $\epsilon$-Decaying, using a decay exponent described on the $x$-axis. }
    \label{fig:decay_sweep}
\end{figure}

\subsection{Rewards}
Clearly, the choice of rewards influences which matchings will have optimal utility. 
Rewards are typically drawn from i.i.d. Gaussian or Bernoulli distributions, although \emph{no assumption on the type of distribution is needed}\footnote{This is because each radar only needs to consider the mean of the distribution, which can be determined with confidence given enough data. }. 
As we discussed earlier, different radar waveforms can be used to maximize different target information. 


Previous work in reinforcement learning for radar \cite{9178313} has shown that a reward function based on a weighted combination of \texttt{SINR} and bandwidth improves detection accuracy. Since each radar in this network will have a fixed bandwidth, we'll use a reward function that encourages high \texttt{SINR}, and discourages collisions: 
\begin{equation}
    \label{eq:rewards}
    \mu_{i,j} = \begin{cases}
    \alpha \left(\texttt{SINR}_{i,j}+\beta\right), & I_i(t)=0\\
    0, & I(t) = 1
    \end{cases}
\end{equation}
where $\alpha$ and $\beta$ are parameters used to map the \texttt{SINR} observed by radar $i$ selection action $j$ roughly to the unit interval. 
$I_i(t)$ indicates whether or not radar $i$ experienced a collision. 
Note that while the \texttt{SINR} appears dependent on both radar and action choice, it is actually only dependent on the action choice $j$, with the radar dependence indicating which radar selected that action. 



In simulations, each node samples $\texttt{SINR}_{i,j}$ from a normal distribution with mean $\frac{\mu_{i,j}}{\alpha} - \beta$. 

\subsubsection{Regret}
MABs have convenient theoretical guarantees on parameters such as regret, which is the difference in reward between actions an agent took and the actions that some oracle with perfect knowledge of rewards would take. 
Regret comes in several flavors. Specifically, we will discuss the so-called ``weak'' regret, which compares actions to the best (on average) actions. 
Regret is defined for the single radar as Eq. (\ref{eq:single_regret}), where $\mu^*$ is the reward of the optimal action. 
\begin{equation}
    \label{eq:single_regret}
    R_t(r) = t\mu^* - \sum_{t_0=1}^T \mu(t_0)
\end{equation}
Regret is also defined for a group of players. 
Formally, after any PRI $1 \leq t \leq T$
\begin{equation}
    \label{eq:regret}
    R_t = tU^* - \sum_{t_0=1}^T\sum_{r\in\mathcal{P}} \mu_{r}(t_0)
\end{equation}
is the cumulative regret at the time horizon $T$ for a group of radars $\mathcal{P}$. 
Since the reward function Eq. (\ref{eq:rewards}) penalizes low \texttt{SINR} and is maximized by high \texttt{SINR}, low regret will also correspond to good radar tracking performance. 

Note that regret only has meaning if the number of available actions is greater than one. 

Finally we can write the \emph{average cumulative regret} in a time step $t\leq T$ as
\begin{equation}
    \overline{R}_t = \frac{R_t}{t}
    \label{eq:ave_cum_regret}
\end{equation}
We use Eq. (\ref{eq:ave_cum_regret}) in the results shown below. 

\subsection{Performance Analysis}
Given the structure of this environment, we can make some claims. 
First, we need to define the \emph{network \texttt{SINR}}. 
\begin{definition}[Network \texttt{SINR}]
The network $\texttt{SINR}_\pi$ is the sum of the \texttt{SINR} at each node $i$, written as $\texttt{SINR}_{i,j}$ which select actions according to the policy $\pi$. 
\begin{equation}
    \texttt{SINR}_\pi = \sum_{r_i,f_j \in \pi} \texttt{SINR}_{i,j}
\end{equation}
\end{definition}

\begin{lemma}
Maximal utility implies maximum network \texttt{SINR}. 
\end{lemma}
\begin{proof}
    First, note that $\alpha\left(\texttt{SINR}+\beta\right) \geq 0$. Trivially, if this quantity is always equal to zero, utility is always maximized. So, without loss of generality, assume that $\alpha\left(\texttt{SINR}+\beta\right) > 0$.

    Since $U(\pi) = U^*$ implies $I(t)=0$ for all radars $r_i$ in the set $\mathcal{P}$ and the matching $\pi$, we can simplify the reward equation to 
    \begin{equation*}
        \mu_{i,j} = \alpha\left(\texttt{SINR}_{i,j} + \beta\right). 
    \end{equation*}
    Further, 
    \begin{align*}
        U^* &= \max_{\pi \in \mathcal{M}} U(\pi(t))\\
        &= \max_{\pi \in \mathcal{M}} \sum_{r_i,f_j \in \pi} \mu_{i,j}\\
        &= \max_{\pi \in \mathcal{M}} \sum_{r_i,f_j \in \pi} \alpha\left(\texttt{SINR}_{i,j}+\beta\right). 
    \end{align*}
    Since $\alpha$ and $\beta$ are constants and $\alpha\left(\texttt{SINR}_{i,j}+\beta\right)$, this implies that for $U(\pi)=U^*$, 
    \begin{align*}
        \texttt{SINR}_\pi &= \max_{\pi \in \mathcal{M}} \sum_{r_i,f_j\in\pi} \texttt{SINR}_{i,j}
    \end{align*}
    which is the maximum network \texttt{SINR}. 
\end{proof}

Since high \texttt{SINR} corresponds to low estimation variance, it's clear that our reinforcement learning algorithm will result in the best-case radar tracking performance for this environment, assuming that the algorithm is capable of converging to the optimal matching.

\section{Simulations}
\begin{table*}[t]
    \centering
    \caption{Simulation parameters, unless stated otherwise. }
    \begin{tabular}{||c | c||c | c||}
    \hline 
    Parameter & Value & Parameter & Value\\
    \hline \hline
    Number of Radars     & 3 & Number of Targets & 1\\
    \hline 
    PRIs per CPI    & 400 & Target Initial Position & [400,400] m\\
    \hline 
    Total CPIs   & 50 & Bandwidth & 20MHz\\
    \hline 
    Typical SNR & 12 dB & Averaged Simulations & 50\\
    \hline
    Frequency & 2.4GHz & PRI Duration & $1.024\times10^{-4}$\\
    \hline
    \end{tabular}
    
    \label{tab:params}
\end{table*}

We will investigate the performance of a radar network which tracks a moving object. 
Each radar is stationary and uses a combined center frequency and waveform selection algorithm (which we specify) to estimate target position, Doppler, range, and angle. 
The individual node locations in the network are modeled as a Binomial Point Process, with a fixed number of nodes placed randomly in the disc centered at $[500, 500]$m with a radius of $500$m. 

Each radar implements a two-dimensional Kalman filter, estimating position $\hat{x} = [x,y]$ and velocity $\hat{v} = [v_x, v_y]$. 
At the end of each CPI, each radar conducts processing on the returns it observed, updates its Kalman filter, and passes the updated track to a fusion center. 
At the fusion center, the position estimates are averaged using equal weights. 
Table \ref{tab:params} lists common simulation parameters. 
%



\emph{Interference}. We model interference through the reward function. 
Each frequency band is assumed to have some interference to noise ratio, with a primary user in one sub-band. Waveforms which overlap with this primary user will experience lower \texttt{SINR} than those which do not. 
Since we'd like to use as much of the available bandwidth as possible to maximize target information, we'll add a term to the reward equation which penalizes waveforms, inversely proportional to their bandwidth.  

%
    
%


%

In the environment considered here, single-node full-band radars will always be outperformed by networks of multiple radar nodes, and static allocation networks will be outperformed by those using coexistence strategies. 
To the latter point, we first present the regret performance shown in Fig. \ref{fig:single_regret}. 
This is a comparison of a network using fixed allocations for frequency and waveform, and one using SAA for waveform selection. 
We can see that the network using sense and avoid for waveform selection is able to obtain far less regret than that using a static allocation, however, the regret is linear since SAA is not optimal. This performance corresponds to the radar tracking error shown in Fig. \ref{fig:single_error}, where the network using SAA has far better performance. 

\begin{figure}[t!]
    \centering
    \includegraphics{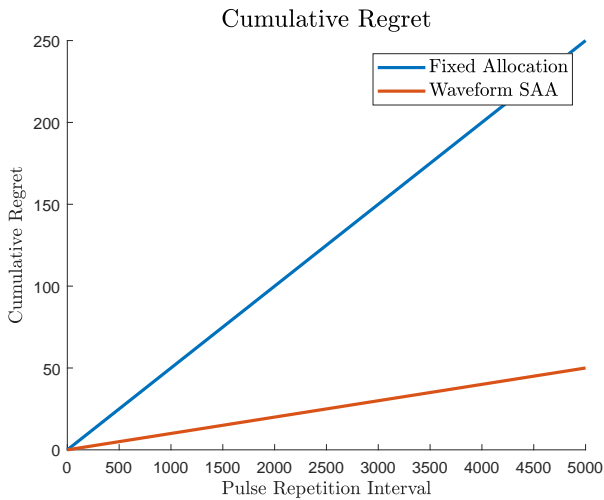}
    \caption{Cumulative regret for two different configurations: the first network of two radars only uses fixed allocations for center frequency and waveform selection, while the second network of two radars uses a fixed allocation for center frequency selection and SAA for waveform selection. }
    \label{fig:single_regret}
\end{figure}
\begin{figure}
    \centering
    \includegraphics{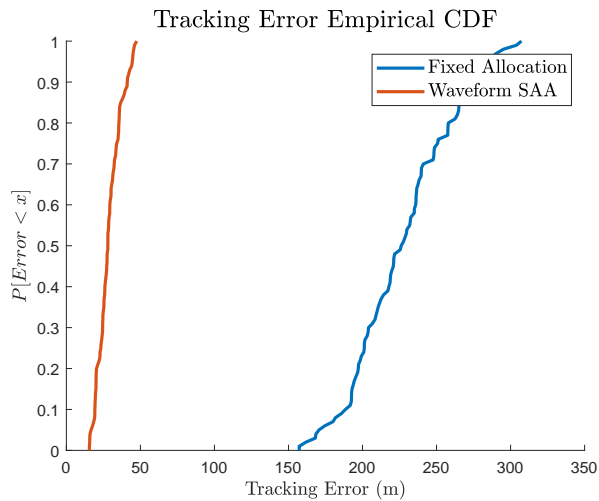}
    \caption{Tracking error for two networks of two radars each. The second network, using SAA for waveform selection, outperforms the network using no intelligent selection. }
    \label{fig:single_error}
\end{figure}

Next, we can look at the benefit of increasing the number of nodes in the network. 
Each network uses the MCTopM algorithm for center frequency selection and $\varepsilon$-Decaying for waveform selection. 
Fig. \ref{fig:netSizeRegret} examines the difference in \emph{average} cumulative regret Eq. (\ref{eq:ave_cum_regret}). 
The networks considered consist of two, three, and four nodes. 
We can see that on average, each node in each of these networks attains the same amount of regret. 
In addition, it's clear that the average cumulative regret is asymptotically going to zero; this means that the center frequency and waveform selections are optimal by roughly the $100^{th}$ PRI. 

Further, in Fig. \ref{fig:netSizeError}, we see that the \emph{tracking error decreases with network size}. 
This is because even though the actions being selected in each network approach the optimal, the averaged tracking estimates over greater network sizes improve. 

\begin{figure}
    \centering
    \includegraphics{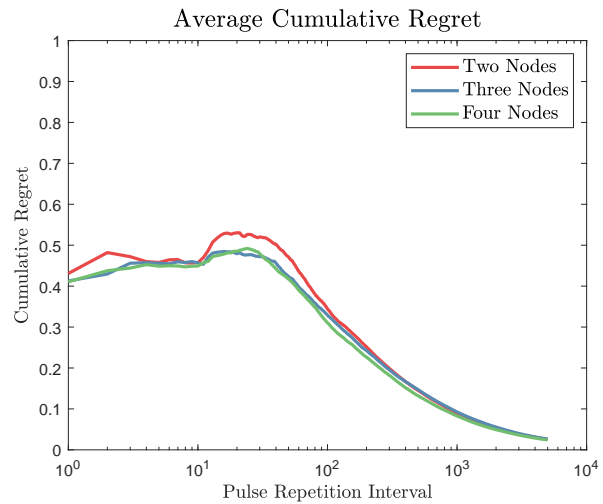}
    \caption{Network regret for networks of two, three, and four radars. Since the average regret \emph{per radar} does not increase with the network size, we can see that there is no impact to the learning problem from additional radar nodes. }
    \label{fig:netSizeRegret}
\end{figure}
\begin{figure}
    \centering
    \includegraphics{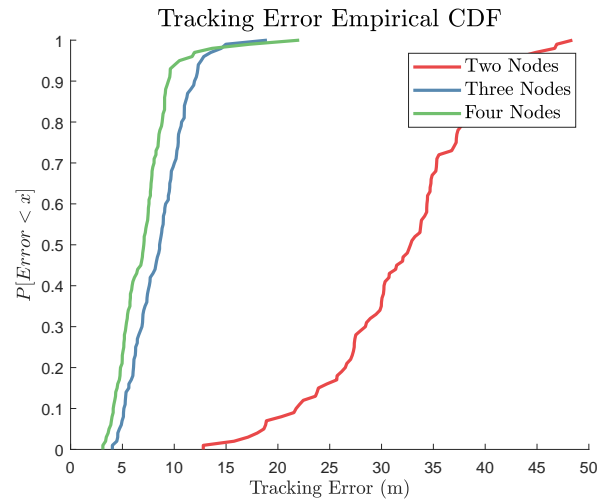}
    \caption{Network tracking performance for network sizes of two, three, and four radars. The improvement from three to four radars is not as pronounced as from two to three, since the observation quality of the fourth radar is less than the others due to the environment configuration. }
    \label{fig:netSizeError}
\end{figure}

Figs. \ref{fig:NetRegret} and \ref{fig:NetError} show the benefits of incorporating different cognitive strategies. 
We can see that using any strategy with MC or MCTopM outperforms the fixed-allocation or waveform SAA strategies shown in Fig. \ref{fig:single_error}. 
This is because even while several radar nodes may be forced into using sub-bands or waveforms with low \texttt{SINR}, the network as a whole can mitigate these errors. 
Note also that the average cumulative regret \ref{fig:netSizeRegret} informs the tracking performance Fig. \ref{fig:netSizeError}. 
Those strategies that do not asymptotically approach zero in average cumulative regret will tend to have higher tracking error than those which approach zero. 
The best-performing algorithm combination can be seen to be MCTopM coupled with $\varepsilon$-Decaying, which is as we would expect. 
The UCB-informed MCTopM is able to balance exploration and exploitation while avoiding mutual interference, and $\varepsilon$-Decaying causes the waveform selection algorithm to converge early in the game.

\begin{figure}
    \centering
    \includegraphics{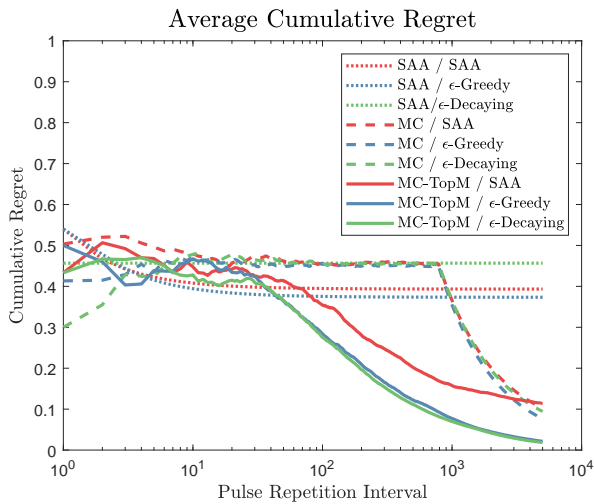}
    \caption{Average cumulative regret for nine different CRNs. Each network has three nodes and employs a two-step algorithm as defined above. Networks employing SAA for frequency selection tend to converge to a sub-optimal solution, while those using MCTopM tend to converge towards the optimal. Similarly, any network using SAA for waveform selection will not obtain optimal performance, and E-Decreasing will have lower regret in each time step. }
    \label{fig:NetRegret}
\end{figure}
\begin{figure}
    \centering
    \includegraphics{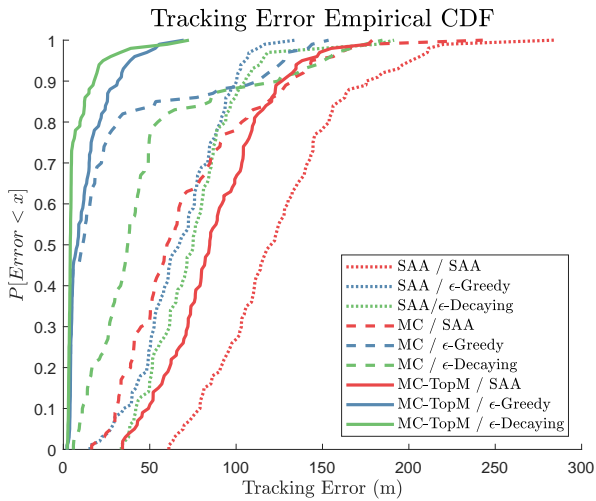}
    \caption{Radar tracking error for these algorithm combinations. As indicated by the regret performance, any networks using SAA for frequency or waveform selection will be outperformed by those using MMAB strategies. The regret performance and specifically the earlier convergence of MCTopM / $\varepsilon$-Decaying allow this combination to obtain lower error on average than other algorithms. }
    \label{fig:NetError}
\end{figure}
\section{Conclusions}
We have seen that radar networks using any amount of cognition for center frequency and waveform selection will outperform pre-allocation techniques, even with equivalent total bandwidth and power distributed through the network.
A major condition for the proper functionality of a radar sensor network is the selection of non-overlapping frequency allocations, which we define as optimal configurations. 
We accomplished this through the use of Multiplayer Multi-Armed Bandit algorithms, which frame this as an iterative decision-making problem. 
We provided a system model which we used to further define a method for the detection of \emph{collisions}, which are instances of two cognitive radar nodes using the same frequency band.

In addition, we described a method of using \emph{single}-player bandit algorithms to accomplish outside interference avoidance. 
We did this by providing an instance of a single-player bandit algorithm to each choice of center frequency, and using the same \texttt{SINR}-based rewards as used for the MMAB frequency selection algorithm. 
The goal of the single-player bandit was to choose from a library of orthogonal waveforms to, over time, avoid any interferers in the environment and obtain the best \texttt{SINR} possible.

All of the algorithms described above operate in a \emph{decentralized} manner, meaning there is no need for information from a central coordinator to make decisions. 
All decisions are made locally to each radar node which reduces communication overhead as well as improves the redundancy of the system to loss of the coordinator.

We demonstrated that the CRN using a combination of MMAB and single-player bandit outperforms SAA as well as the static, pre-allocated case. 
Specifically, MCTopM for center frequency selection paired with $\varepsilon$-Decaying is shown to obtain the lowest regret bound, which we prove corresponds to optimal asymptotic radar tracking performance.

Future work will include a stochastic treatment of the collision detection problem. 
While in this work we use the Nearest Neighbor distribution function of the BPP to estimate the power received from LoS mutual interference, a more rigorous analysis could be accomplished. 
Specifically we assume that the radars point in uniformly random directions in each time step while developing the collision detection algorithm, while in reality each node may be able to form a prior distribution on the pointing of other nodes given a target's estimated position. 
This could be used to inform a collision detection algorithm.

In addition, in this work we assume that all decision-making can be accomplished on the order of a PRI, while in many radar implementations, this and all other processing occurs on the CPI scale to preserve coherence pulse-to-pulse. 
While there exist methods to do coherent processing on pulses which vary inside the CPI, these remain somewhat impractical. 
Future work may examine the impact of reducing all decision-making to the CPI scale, and analyzing the consequences. 
The current algorithms would be hampered by this change, as observations of the environment are assumed to immediately follow the previous action, and immediately precede the next action. 
In addition more realistic radar scenarios such as clutter-limited environments and multiple targets may be considered.

\bibliographystyle{IEEEtran}
\bibliography{bibli}

\end{document}